\documentclass[hidelinks,11pt]{article}

\usepackage{amsmath}
\usepackage{amssymb}
\usepackage{amsfonts,amsthm}
\usepackage{thmtools,thm-restate}
\usepackage{tabularx,lipsum,environ}
\usepackage{multirow}
\usepackage{hyperref}
\usepackage{graphicx,float}
\usepackage{enumitem,linegoal}
\usepackage{caption}
\usepackage{subcaption}
\usepackage{complexity}

\usepackage{enumitem}

\setlist[itemize]{noitemsep}
\setlist[enumerate]{noitemsep}

\usepackage{authblk}

\usepackage{boxedminipage,tikz}

\newtheorem{theorem}{Theorem}[section]
\newtheorem{lemma}[theorem]{Lemma}

\newtheorem{corollary}[theorem]{Corollary}
\newtheorem{Claim}{Claim}
\newtheorem{fact}[theorem]{Fact}

\newtheorem{observation}{Observation}

\theoremstyle{definition}

\newenvironment{claimproof}{\begin{proof}\renewcommand{\qedsymbol}{\claimqed}}{\end{proof}\renewcommand{\qedsymbol}{\plainqed}}
\let\plainqed\qedsymbol

\usetikzlibrary{arrows}
\usetikzlibrary{arrows.meta}
\usetikzlibrary{decorations.pathmorphing}
\usetikzlibrary{decorations.markings}
\usetikzlibrary{calc}
\tikzset{
  circ/.style = {circle,draw,fill,inner sep=1.3pt},
  circr/.style = {circle,draw=red,fill=red,inner sep=1.3pt},
  scirc/.style = {circle,draw,fill,inner sep=.5pt},
  invisible/.style = {draw=none,inner sep=0pt,font=\tiny},
  nonedge/.style={decorate,decoration={snake,amplitude=.3mm,segment length=1mm},draw}
}

\usetikzlibrary{arrows.meta}
\tikzset{
  circ/.style = {circle,draw,fill,inner sep=1.3pt},
  invisible/.style = {draw=none,inner sep=0pt,font=\tiny}
}

\usepackage[a4paper,
            top=1.1in, bottom=1.4in, left=1.1in, right=1.1in,
            footskip=.5in]{geometry}

\newcommand\yes{\textsc{Yes}}
\newcommand\no{\textsc{No}}

\newcommand\dom{\textsc{Dominating Set}}
\newcommand\contracd{\textsc{1-Edge Contraction($\gamma$)}}

\newcommand\kcontracd{\textsc{$k$-Edge Contraction($\gamma$)}}

\title{Blocking dominating sets for $H$-free graphs via edge contractions}

\author[1]{Esther Galby}
\author[2]{Paloma T. Lima}
\author[1]{Bernard Ries}

\affil[1]{University of Fribourg\\
	Fribourg, Switzerland}

\affil[2]{University of Bergen\\
	Bergen, Norway}

\date{}

\begin{document}

\maketitle

\begin{abstract}
In this paper, we consider the following problem: given a connected graph $G$, can we reduce the domination number of $G$ by one by using only one edge contraction? We show that the problem is $\mathsf{NP}$-hard when restricted to $\{P_6,P_4+P_2\}$-free graphs and that it is $\mathsf{coNP}$-hard when restricted to subcubic claw-free graphs and $2P_3$-free graphs. As a consequence, we are able to establish a complexity dichotomy for the problem on $H$-free graphs when $H$ is connected. 
\end{abstract}

\section{Introduction}
\label{s-intro}

A {\it blocker problem} asks whether given a graph $G$, a graph parameter $\pi$, a set $\mathcal{O}$ of one or more graph operations and an integer $k \geq 1$, $G$ can be transformed into a graph $G'$ by using at most $k$ operations from $\mathcal{O}$ such that $\pi(G') \leq \pi(G) - d$ for some {\it threshold} $d \geq 0$. Such a designation follows from the fact that the set of vertices or edges involved can be viewed as ''blocking'' the parameter $\pi$. Identifying such sets may provide information on the structure of the input graph; for instance, if $\pi = \alpha$, $k=d=1$ and $\mathcal{O} = \{$vertex deletion$\}$, the problem is equivalent to testing whether the input graph contains a vertex that is in every maximum independent set (see \cite{paulusma2017blocking}). Blocker problems have received much attention in the recent literature (see for instance \cite{BTT11,bazgan2013critical,RBPDCZ10,Bentz,CWP11,DPPR15,diner2018contraction,contracdom,keller2018blockers,keller2013blockers,PBP,nasirian2019exact,pajouh2015minimum,PPR16,paulusma2017blocking,paulusma2018critical}) and have been related to other well-known graph problems such as \textsc{Hadwiger Number}, \textsc{Club Contraction} and several graph transversal problems (see for instance \cite{DPPR15,PPR16}). The graph parameters considered so far in the literature are the chromatic number, the independence number, the clique number, the matching number and the vertex cover number while the set $\mathcal{O}$ is a singleton consisting of a vertex deletion, edge contraction, edge deletion or edge addition. In this paper, we focus on the domination number $\gamma$, let $\mathcal{O}$ consists of an edge contraction and set the threshold $d$ to one.

Formally, let $G=(V,E)$ be a graph. The {\it contraction} of an edge $uv\in E$ removes vertices $u$ and $v$ from $G$ and replaces them by a new vertex that is made adjacent to precisely those vertices that were adjacent to $u$ or $v$ in $G$ (without introducing self-loops nor multiple edges). We say that a graph $G$ can be \textit{$k$-contracted} into a graph~$G'$, if $G$ can be transformed into $G'$ by a sequence of at most~$k$ edge contractions, for an integer $k\geq 1$. The problem we consider is then the following (note that contracting an edge cannot increase the domination number).

\begin{center}
\begin{boxedminipage}{.99\textwidth}
$k$-\textsc{Edge Contraction($\gamma$)}\\
\begin{tabular}{ r p{0.8\textwidth}}
\textit{~~~~Instance:} &A connected graph $G=(V,E)$.\\
\textit{Question:} &Can $G$ be $k$-contracted into a graph $G'$ such that $\gamma(G')~\leq~\gamma(G) -1$?
\end{tabular}
\end{boxedminipage}
\end{center}

Reducing the domination number using edge contractions was first considered in \cite{HX10}. The authors proved that for a connected graph $G$ such that $\gamma(G)\geq 2$, we have $ct_\gamma (G)\leq 3$, where $ct_{\gamma}(G)$ denotes the minimum number of edge contractions required to transform $G$ into a graph $G'$ such that $\gamma (G') \leq \gamma (G) -1$ (note that if $\gamma(G) = 1$ then $G$ is a \no-instance for \kcontracd{} independently of the value of $k$). Thus, if $G$ is a connected graph with $\gamma (G) \geq 2$, then $G$ is always a \yes-instance for \kcontracd{} when $k \geq 3$. It was later shown in \cite{contracdom} that \kcontracd{} is $\mathsf{coNP}$-hard for $k \leq 2$ and so, restrictions on the input graph to some special graph classes were considered. In particular, the authors in \cite{contracdom} proved that for $k=1,2$, the problem is polynomial-time solvable for $P_5$-free graphs while for $k=1$, it remains $\mathsf{NP}$-hard when restricted to $P_9$-free graphs and $\{C_3,\ldots,C_\ell\}$-free graphs, for any $\ell \geq 3$. 

In this paper, we continue the systematic study of the computational complexity of \contracd{} initiated in \cite{contracdom}. Ultimately, the aim is to obtain a complete classification for \contracd{} restricted to $H$-free graphs, for any (not necessarily connected) graph $H$, as it has been done for other blocker problems (see for instance \cite{diner2018contraction,paulusma2017blocking,paulusma2018critical}). As a step towards this end, we prove the following three theorems.

\begin{theorem}
\label{thm:p6free}
\contracd{} is $\mathsf{NP}$-hard when restricted to $\{P_6,P_4+P_2\}$-free graphs.
\end{theorem}

\begin{theorem}
\label{thm:clawfree}
\contracd{} is $\mathsf{coNP}$-hard when restricted to subcubic claw-free graphs.
\end{theorem}

\begin{theorem}
\label{thm:2p3free}
\contracd{} is $\mathsf{coNP}$-hard when restricted to $2P_3$-free graphs.
\end{theorem}

Theorems~\ref{thm:p6free} and~\ref{thm:clawfree} lead to a complexity dichotomy for $H$-free graphs when $H$ is connected. Indeed, since \contracd{} is $\mathsf{NP}$-hard when restricted to $\{C_3,\ldots,C_\ell\}$-free graphs, for any $\ell \geq 3$, it follows that \contracd{} is $\mathsf{NP}$-hard for $H$-free graphs when $H$ contains a cycle. If $H$ is a tree with a vertex of degree at least three, we conclude by Theorem \ref{thm:clawfree} that \contracd{} is $\mathsf{coNP}$-hard for $H$-free graphs; Theorem~\ref{thm:p6free} shows that if $H$ is a path of length at least 6, then \contracd{} is $\mathsf{NP}$-hard for $H$-free graphs. Finally, since in~\cite{contracdom} \contracd{} is shown to be polynomial-time solvable on $\{P_5+pK_1\}$-free graphs for any $p\geq 0$, then \contracd{} is polynomial-time solvable on $H$-free graphs if $H\subseteq_i P_5$. We therefore obtain the following result.

\begin{corollary}
Let $H$ be a connected graph.  If $H\subseteq_i P_5$ then \contracd{} is polynomial-time solvable on $H$-free graphs, otherwise it is $\mathsf{NP}$-hard or $\mathsf{coNP}$-hard.
\end{corollary}

If the graph $H$ is not required to be connected, we know the following. As previously mentioned, \contracd{} is $\mathsf{NP}$-hard (resp.\ $\mathsf{coNP}$-hard) on $H$-free graphs when $H$ contains a cycle (resp.\ an induced claw). Thus, there remains to consider the case where $H$ is a linear forest. Theorems~\ref{thm:p6free} and~\ref{thm:2p3free} show that if $H$ contains either a $P_6$, a $P_4+P_2$ or a $2P_3$ as an induced subgraph, then \contracd{} is $\mathsf{NP}$-hard or $\mathsf{coNP}$-hard on $H$-free graphs. Since it is known that \contracd{} is polynomial-time solvable on $H$-free graphs if $H\subseteq_i P_5+pK_1$, there remains to determine the complexity status of the problem restricted to $H$-free graphs when $H=P_3+qP_2+pK_1$, for $q\geq 1$ and $p\geq 0$.


\section{Preliminaries}
\label{s-pre}

Throughout the paper, we only consider finite, undirected, connected graphs that have no self-loops or multiple edges. We refer the reader to~\cite{Di05} for any terminology and notation not defined here.

For $n\geq 1$, the path and cycle on $n$ vertices are denoted by $P_n$ and $C_n$ respectively. The \textit{claw} is the complete bipartite graph with one partition of size one and the other of size three.

Let $G=(V,E)$ be a graph and let $u\in V$. We denote by $N_G(u)$, or simply $N(u)$ if it is clear from the context, the set of vertices that are adjacent to $u$ i.e., the {\it neighbors} of $u$, and let $N[u]=N(u)\cup \{u\}$. The \textit{degree} of a vertex $u$, denoted by $d_G(u)$ or simply $d(u)$ if it is clear from the context, is the size of its neighborhood i.e., $d(u) = \vert N(u) \vert$. The maximum degree in $G$ is denoted by $\Delta (G)$ and $G$ is \textit{subcubic} if $\Delta (G) \leq 3$. 

For a family $\{H_1,\ldots,H_p\}$ of graphs, $G$ is said to be {\it $\{H_1,\ldots,H_p\}$-free} if $G$ has no induced subgraph isomorphic to a graph in $\{H_1,\ldots,H_p\}$; if $p=1$ we may write $H_1$-free instead of $\{H_1\}$-free.  For a subset $V'\subseteq V$, we let $G[V']$ denote the subgraph of $G$ {\it induced} by $V'$, which has vertex set~$V'$ and edge set $\{uv\in E\; |\; u,v\in V'\}$.

A subset $S \subseteq V$ is called an {\it independent set} or is said to be \textit{independent}, if no two vertices in $S$ are adjacent. A subset $D\subseteq V$ is called a {\it dominating set}, if every vertex in $V\setminus D$ is adjacent to at least one vertex in $D$; the {\it domination number} $\gamma(G)$ is the number of vertices in a minimum dominating set. For any $v \in D$ and $u \in N[v]$, $v$ is said to \textit{dominate} $u$ (in particular, $v$ dominates itself). We say that \textit{$D$ contains an edge} (or more) if the graph $G[D]$ contains an edge (or more). A dominating set $D$ of $G$ is \textit{efficient} if for every vertex $v \in V$, $\vert N[v] \cap D \vert = 1$ that is, $v$ is dominated by exactly one vertex.

In the following, we consider those graphs for which one contraction suffices to decrease their domination number by one. A characterization of this class is given in \cite{HX10}.

\begin{theorem}[\cite{HX10}]
\label{theorem:contracdom}
For a connected graph $G$, $ct_\gamma (G)=1$ if and only if there exists a minimum dominating set in $G$ that is not independent.
\end{theorem}

In order to prove Theorems \ref{thm:clawfree} and \ref{thm:2p3free}, we introduce to two following problems.

\begin{center}
\begin{boxedminipage}{.99\textwidth}
\textsc{\sc All Efficient MD}\\[2pt]
\begin{tabular}{ r p{0.8\textwidth}}
\textit{~~~~Instance:} &A connected graph $G=(V,E)$.\\
\textit{Question:} &Is every minimum dominating set of $G$ efficient?
\end{tabular}
\end{boxedminipage}
\end{center}

\begin{center}
\begin{boxedminipage}{.99\textwidth}
\textsc{\sc All Independent MD}\\[2pt]
\begin{tabular}{ r p{0.8\textwidth}}
\textit{~~~~Instance:} &A connected graph $G=(V,E)$.\\
\textit{Question:} &Is every minimum dominating set of $G$ independent?
\end{tabular}
\end{boxedminipage}
\end{center}

The following is then a straightforward consequence of Theorem \ref{theorem:contracdom}.

\begin{fact}
\label{obs:equi}
Given a graph $G$, $G$ is a \yes-instance for \contracd{} if and only if $G$ is a \no-instance for {\sc All Independent MD}.
\end{fact}


\section{The proof of Theorem~\ref{thm:p6free}}

In this section, we show that \contracd{} is $\mathsf{NP}$-hard when restricted to $\{P_6,P_4+P_2\}$-free graphs. \\

To this end, we give a reduction from \dom. Given an instance $(G,\ell)$ for \dom{}, we construct an instance $G'$ for \contracd{} as follows. We denote by $\{v_1,\ldots,v_n\}$ the vertex set of $G$. The vertex set of the graph $G'$ is given by $V(G')=V_0\cup\ldots\cup V_\ell\cup\{x_0,\ldots,x_\ell,y\}$, where each $V_i$ is a copy of the vertex set of $G$. We denote the vertices of $V_i$ by $\{v^1_i,v^2_i,\ldots,v^n_i\}$. The adjacencies in $G'$ are defined as follows:
\begin{itemize}
\item[-] $V_0\cup\{x_0\}$ is a clique;
\item[-] $yx_0\in E(G')$;
\end{itemize}

 For $1\leq i\leq \ell$:
\begin{itemize}
\item[-] $V_i$ is an independent set;
\item[-] $x_i$ is adjacent to all the vertices of $V_0\cup V_i$;
\item[-] $v^j_i$ is adjacent to $\{v^a_0~|~v_a\in N_G[v_j]\}$.
\end{itemize}

\begin{figure}[htb]
\centering
\begin{tikzpicture}[scale=.5]
\draw (0,0) ellipse (2cm and .8cm);
\node[draw=none] at (0,0) {$V_0$};

\node[circ,label=above left:{\small $x_1$}] (x1) at (2.21,1.27) {};
\draw[rotate=37] (4.5,0) ellipse (1cm and .5cm);
\node[draw=none] at (3.59,2.71) {$V_1$};
\draw[very thick] (x1) -- (1.99,.07)
(x1) -- (.17,.8)
(x1) -- (2.74,2.34)
(x1) -- (3.18,1.99);

\node[circ,label=below left:{\small $x_2$}] (x2) at (2.21,-1.27) {};
\draw[rotate=-37] (4.5,0) ellipse (1cm and .5cm);
\node[draw=none] at (3.59,-2.71) {$V_2$};
\draw[very thick] (x2) -- (1.99,-.07)
(x2) -- (.17,-.8)
(x2) -- (2.74,-2.34)
(x2) -- (3.18,-1.99);

\node[draw=none] at (0,-1.8) {$\ldots$};

\node[circ,label=below right:{\small $x_{\ell}$}] (xl) at (-2.21,-1.27) {};
\draw[rotate=37] (-4.5,0) ellipse (1cm and .5cm);
\node[draw=none] at (-3.59,-2.71) {$V_{\ell}$};
\draw[very thick] (xl) -- (-1.99,-.07)
(xl) -- (-.17,-.8)
(xl) -- (-2.74,-2.34)
(xl) -- (-3.18,-1.99);

\node[circ,label=above right:{\small $x_0$}] (xl1) at (-2.21,1.27) {};
\node[circ,label=left:{\small $y$}] (y) at (-3,2) {};
\draw[very thick] (xl1) -- (-1.99,.07)
(xl1) -- (-.17,.8);
\draw[-] (xl1) -- (y);
\end{tikzpicture}
\caption{The graph $G'$ (thick lines indicate that the vertex $x_i$ is adjacent to every vertex in $V_0$ and $V_i$, for $i=0,\ldots,\ell$).}
\label{fig:dom1ec}
\end{figure}
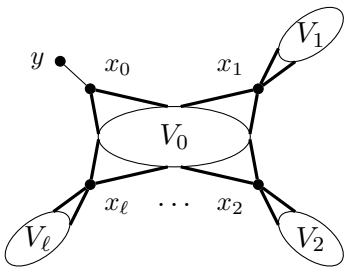

\begin{Claim}
\label{claim:gammaG'}
$\gamma (G') = \min \{\gamma (G) + 1, \ell + 1\}$. 
\end{Claim}

It is clear that $\{x_0,x_1, \ldots, x_\ell\}$ is a dominating set of $G'$; thus, $\gamma (G') \leq \ell +1$. If $\gamma (G) \leq \ell$ and $\{v_{i_1},\ldots, v_{i_k}\}$ is a minimum dominating set of $G$, it is easily seen that $\{v^0_{i_1},\ldots, v^0_{i_k},x_0\}$ is a dominating set of $G'$. Thus, $\gamma (G') \leq \gamma (G) + 1$ and so, $\gamma (G') \leq \min \{\gamma (G) + 1, \ell + 1\}$. Now, suppose to the contrary that $\gamma (G') < \min \{\gamma (G) + 1, \ell + 1\}$ and consider a minimum dominating set $D'$ of $G'$. We first make the following simple observation.

\begin{observation}
\label{obs:yxl+1}
For any dominating set $D$ of $G'$, $D \cap \{y,x_0\} \neq~\emptyset$.
\end{observation}

Now, since $\gamma (G') < \ell + 1$, there exists $1 \leq i \leq \ell$ such that $x_i \not\in D'$ (otherwise, $\{x_1,\ldots,x_{\ell}\} \subset D'$ and combined with Observation \ref{obs:yxl+1}, $D'$ would be of size at least $\ell + 1$). But then, $D'' = D' \cap V_0$ must dominate every vertex in $V_i$, and so $\vert D'' \vert \geq \gamma(G)$. Since $\vert D'' \vert \leq \vert D' \vert - 1$ (recall that $D' \cap \{y,x_0\} \neq \emptyset$), we then have $\gamma (G) \leq \vert D' \vert - 1$, a contradiction. Thus, $\gamma (G') = \min \{\gamma (G) + 1, \ell + 1\}$.\\

We now show that $(G,\ell)$ is a \yes-instance for \dom{} if and only if $G'$ is a \yes-instance for \contracd{}.

First assume that $\gamma (G) \leq \ell$. Then, $\gamma (G') = \gamma(G) + 1$ by the previous claim, and if $\{v_{i_1},\ldots,v_{i_k}\}$ is a minimum dominating set of $G$, then $\{v^0_{i_1},\ldots,v^0_{i_k},x_0\}$ is a minimum dominating set of $G'$ which is not stable. Hence, by Theorem \ref{theorem:contracdom}, $G'$ is a \yes-instance for \contracd{}.

Conversely, assume that $G'$ is a \yes-instance for \contracd{} i.e., there exists a minimum dominating set $D'$ of $G'$ which is not stable (see Theorem \ref{theorem:contracdom}). Then, Observation \ref{obs:yxl+1} implies that there exists $1 \leq i \leq \ell$ such that $x_i \not\in D'$; indeed, if it weren't the case, then by Claim \ref{claim:gammaG'} we would have $\gamma (G') = \ell + 1$ and thus, $D'$ would consist of $x_1, \ldots, x_{\ell}$ and either $y$ or $x_0$. In both cases, $D'$ would be stable, a contradiction. It follows that $D'' = D' \cap V_0$ must dominate every vertex in $V_i$ and thus, $\vert D'' \vert \geq \gamma (G)$. But $\vert D'' \vert \leq \vert D' \vert -1$ (recall that $D' \cap \{y,x_0\} \neq \emptyset$) and so by Claim \ref{claim:gammaG'}, $\gamma (G) \leq \vert D' \vert - 1 \leq (\ell + 1) - 1$ that is, $(G,\ell)$ is a \yes-instance for \dom{}.\\

In the following, we show that $G'$ is a $P_6$-free graph. Let $P$ be an induced path of $G'$. We start by observing that, since $V_0$ is a clique, $|V(P)\cap V_0|\leq 2$. If $|V(P)\cap V_0|= 0$, since each $V_i$ is independent and the same holds for $\{x_0,\ldots,x_\ell\}$, we have that $|V(P)|\leq 3$. We now consider the following two cases:\\

\noindent \textbf{Case 1.} $|V(P)\cap V_0|= 2$. 

Let $u,v\in V_0$ be the vertices of $V(P)\cap V_0$. Since $P$ is an induced path, $u$ and $v$ appear consecutively in $P$, that is, $uv\in E(P)$. Furthermore, $V(P)\cap \{x_0,\ldots,x_\ell\}=\emptyset$, since $u$ and $v$ are adjacent to all the vertices of $\{x_0,\ldots,x_\ell\}$. Let $w\in V_i$ be the other neighbor of $u$ in $P$. Since $N(w)\subset V_0\cup \{x_i\}$, $w$ can have no neighbor in $P$ other than $u$, that is, $w$ is an endpoint of the path. Symmetrically, the same holds for a neighbor of $v$ that is different from $u$. Hence, we conclude that $|V(P)|\leq 4$.\\

\noindent \textbf{Case 2.} $|V(P)\cap V_0|= 1$.

Let $u\in V_0$ be the vertex of $V(P)\cap V_0$. If $V(P)\cap \{x_0,\ldots,x_\ell\}=\emptyset$, then it is easy to see that $|V(P)|\leq 3$, since any neighbor of $u$ in the path must belong to $\cup_{1\leq i\leq\ell} V_i$ and, by the same argument used in Case 1, such a neighbor would have to be an endpoint of the path. If $V(P)\cap \{x_0,\ldots,x_\ell\}\neq\emptyset$, let $x_i$ be a vertex that is in $P$. Since $ux_i\in E(G')$, we necessarily have that $ux_i\in E(P)$. Let $w$ be the other neighbor of $x_i$ in $P$. Then $w\in V_i$, since $N(x_i)=V_0\cup V_i$. By the same argument used above, $w$ must then be an endpoint of the path; and since $u$ can have at most two neighbors in $\{x_0,\ldots,x_\ell\}$, we conclude that $|V(P)|\leq 5$.\\

To see that $G'$ is also a $\{P_4+P_2\}$-free graph it suffices to note that any induced $P_4$ of $G'$ contains at least one vertex of the clique $V_0$. This concludes the proof of Theorem~\ref{thm:p6free}.


\section{The proof of Theorem \ref{thm:clawfree}}

\setcounter{observation}{0}
\setcounter{Claim}{0}

In this section, we show that \contracd{} is $\mathsf{coNP}$-hard when restricted to subcubic claw-free graphs. To this end, we first prove the following.

\begin{lemma}
\label{lemma:efficient}
{\sc All Efficient MD} is $\mathsf{NP}$-hard when restricted to subcubic graphs.
\end{lemma}

\begin{proof}
We reduce from {\sc Positive Exactly 3-Bounded 1-In-3 3-Sat}, where each variable appears in exactly three clauses and only positively, each clause contains three positive literals, and we want a truth assignment such that each clause contains exactly one true literal. This problem is shown to be $\mathsf{NP}$-complete in \cite{moore}. Given an instance $\Phi$ of this problem, with variable set $X$ and clause set $C$, we construct an equivalent instance of {\sc All Efficient MD} as follows. For any variable $x \in X$, we introduce a copy of $C_9$, which we denote by $G_x$, with three distinguished \textit{true vertices} $T^1_x$, $T^2_x$ and $T^3_x$, and three distinguished \textit{false vertices} $F^1_x$, $F^2_x$ and $F^3_x$ (see Fig. \ref{fig:vargad}). For any clause $c \in C$ containing variables $x_1$, $x_2$ and $x_3$, we introduce the gadget $G_c$ depicted in Fig. \ref{fig:clausegad} which has one distinguished \textit{clause vertex} $c$ and three distinguished \textit{variable vertices} $x_1$, $x_2$ and $x_3$ (note that $G_c$ is not connected). For every $j \in \{1,2,3\}$, we then add an edge between $x_j$ and $F_{x_j}^i$ and between $c$ and $T_{x_j}^i$ for some $i \in \{1,2,3\}$ so that $F_{x_j}^i$ (resp. $T_{x_j}^i$) is adjacent to exactly one variable vertex (resp. clause vertex). We denote by $G_{\Phi}$ the resulting graph. Note that $\Delta (G_{\Phi}) = 3$.

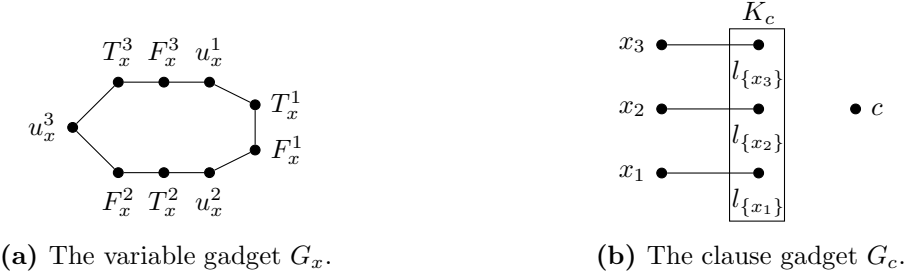
\begin{figure}[htb]
\centering
\begin{subfigure}[b]{.45\textwidth}
\centering
\begin{tikzpicture}[scale=.6]
\node[circ,label=below:{\small $F_x^2$}] (f2) at (1,0) {};
\node[circ,label=below:{\small $T_x^2$}] (t2) at (2,0) {};
\node[circ,label=below:{\small $u_x^2$}] (u2) at (3,0) {};
\node[circ,label=right:{\small $F_x^1$}] (f1) at (4,.5) {};
\node[circ,label=right:{\small $T_x^1$}] (t1) at (4,1.5) {};
\node[circ,label=above:{\small $u_x^1$}] (u1) at (3,2) {};
\node[circ,label=above:{\small $F_x^3$}] (f3) at (2,2) {};
\node[circ,label=above:{\small $T_x^3$}] (t3) at (1,2) {};
\node[circ,label=left:{\small $u_x^3$}] (u3) at (0,1) {};

\draw[-] (u3) -- (f2) -- (u2) -- (f1) -- (t1) -- (u1) -- (t3) -- (u3);
\end{tikzpicture}
\caption{The variable gadget $G_x$.}
\label{fig:vargad}
\end{subfigure}
\hspace*{.5cm}
\begin{subfigure}[b]{.45\textwidth}
\centering
\begin{tikzpicture}[scale=.85]
\node[circ,label=left:{\small $x_1$}] (x1) at (0,1) {};
\node[circ,label=left:{\small $x_2$}] (x2) at (0,2) {};
\node[circ,label=left:{\small $x_3$}] (x3) at (0,3) {};
\node[circ,label=below:{\small $l_{\{x_1\}}$}] (l1) at (1.5,1) {};
\node[circ,label=below:{\small $l_{\{x_2\}}$}] (l2) at (1.5,2) {};
\node[circ,label=below:{\small $l_{\{x_3\}}$}] (l3) at (1.5,3) {};
\node[circ,label=right:{\small $c$}] at (3,2) {};

\draw[-] (x1) -- (l1)
(x2) -- (l2)
(x3) -- (l3);

\draw (1.05,.25) rectangle (1.9,3.25);
\node[draw=none] at (1.5,3.5) {\small $K_c$};
\end{tikzpicture}
\caption{The clause gadget $G_c$.}
\label{fig:clausegad}
\end{subfigure}
\caption{Construction of the graph $G_{\Phi}$ (the rectangle indicates that the corresponding set of vertices induces a clique).}
\end{figure}

\begin{observation}
\label{obs:size}
For any dominating set $D$ of $G_{\Phi}$, $\vert D \cap V(G_x) \vert \geq 3$ for any $x \in X$ and $\vert D \cap V(G_c) \vert \geq 1$ for any $c \in C$. In particular, $\gamma(G_{\Phi}) \geq 3 \vert X \vert + \vert C \vert$.
\end{observation}

Indeed, for any $x \in X$, since $u_x^1$, $u_x^2$ and $u_x^3$ must be dominated and their neighborhoods are pairwise disjoint and contained in $G_x$, it follows that $\vert D \cap V(G_x) \vert \geq 3$. For any $c \in C$, since the vertices of $K_c$ must be dominated and their neighborhoods are contained in $G_c$, $\vert D \cap V(G_c) \vert \geq 1$. $\diamond$ 

\begin{observation}
\label{obs:mingx}
For any $x \in X$, if $D$ is a minimum dominating set of $G_x$ then either $D = \{u_x^1,u_x^2,u_x^3\}$, $D = \{T_x^1,T_x^2,T_x^3\}$ or $D = \{F_x^1,F_x^2,F_x^3\}$.
\end{observation}

\begin{Claim}
\label{clm:phisat}
$\Phi$ is satisfiable if and only if $\gamma (G_{\Phi}) = 3 \vert X \vert + \vert C \vert$.
\end{Claim}

\begin{claimproof}
Assume that $\Phi$ is satisfiable and consider a truth assignment satisfying $\Phi$. We construct a dominating set $D$ of $G_{\Phi}$ as follows. For any variable $x \in X$, if $x$ is true, add $T_x^1$, $T_x^2$ and $T_x^3$ to $D$; otherwise, add $F_x^1$, $F_x^2$ and $F_x^3$ to $D$. For any clause $c \in C$ containing variables $x_1$, $x_2$ and $x_3$, exactly one variable is true, say $x_1$ without loss of generality; we then add $l_{\{x_1\}}$ to $D$. Clearly, $D$ is dominating and we conclude by Observation~\ref{obs:size} that $\gamma (G_{\Phi}) = 3 \vert X \vert + \vert C \vert$.

Conversely, assume that $\gamma (G_{\Phi}) = 3 \vert X \vert + \vert C \vert$ and consider a minimum dominating set $D$ of $G_{\Phi}$. Then by Observation \ref{obs:size}, $\vert D \cap V(G_x) \vert = 3$ for any $x \in X$ and $\vert D \cap V(G_c) \vert = 1$ for any $c \in C$. Now, for a clause $c \in C$ containing variables $x_1$, $x_2$ and $x_3$, if $D \cap \{c,x_1,x_2,x_3\} \neq \emptyset$ then $D \cap V(K_c) = \emptyset$ and so, at least two vertices from $K_c$ are not dominated; thus, $D \cap \{c,x_1,x_2,x_3\} = \emptyset$. It follows that for any $x \in X$, $D \cap V(G_x)$ is a minimum dominating set of $G_x$ which by Observation \ref{obs:mingx} implies either $\{T_x^1,T_x^2,T_x^3\} \subset D$ or $D \cap \{T_x^1,T_x^2,T_x^3\} = \emptyset$; and we conclude similarly that either $\{F_x^1,F_x^2,F_x^3\} \subset D$ or $D \cap \{F_x^1,F_x^2,F_x^3\} = \emptyset$. Now given a clause $c \in C$ containing variables $x_1$, $x_2$ and $x_3$, since $D \cap \{c,x_1,x_2,x_3\} = \emptyset$, at least one true vertex adjacent to the clause vertex $c$ must belong to $D$, say $T_{x_1}^i$ for some $i \in\{1,2,3\}$ without loss of generality. It then follows that $\{T_{x_1}^1,T_{x_1}^2,T_{x_1}^3\} \subset D$ and $D \cap \{F_{x_1}^1,F_{x_1}^2,F_{x_1}^3\} = \emptyset$ which implies that $l_{\{x_1\}} \in D$ (either $x_1$ or a vertex from $K_c$ would otherwise not be dominated). But then, since $x_j$ for $j \neq 1$, must be dominated, it follows that $\{F_{x_j}^1,F_{x_j}^2,F_{x_j}^3\} \subset D$. We thus construct a truth assignment satisfying $\Phi$ as follows: for any variable $x \in X$, if $\{T_x^1,T_x^2,T_x^3\} \subset D$, set $x$ to true, otherwise set $x$ to false.
\end{claimproof} 

\begin{Claim}
\label{clm:eff}
$\gamma (G_{\Phi}) = 3 \vert X \vert + \vert C \vert$ if and only if every minimum dominating set of $G_{\Phi}$ is efficient.
\end{Claim}

\begin{claimproof}
Assume that $\gamma (G_{\Phi}) = 3 \vert X \vert + \vert C \vert$ and consider a minimum dominating set $D$ of $G_{\Phi}$. Then by Observation~\ref{obs:size}, $\vert D \cap V(G_x) \vert = 3$ for any $x \in X$ and $\vert D \cap V(G_c) \vert = 1$ for any $c \in C$. As shown previously, it follows that for any clause $c \in C$ containing variables $x_1$, $x_2$ and $x_3$, $D \cap \{c,x_1,x_2,x_3\} = \emptyset$; and for any $x \in X$, either $\{T_x^1,T_x^2,T_x^3\} \subset D$ or $D \cap \{T_x^1,T_x^2,T_x^3\} = \emptyset$ (we conclude similarly with $\{F_x^1,F_x^2,F_x^3\}$ and $\{u_x^1,u_x^2,u_x^3\}$). Thus, for any $x \in X$, every vertex in $G_x$ is dominated by exactly one vertex. Now given a clause $c \in C$ containing variables $x_1$, $x_2$ and $x_3$, since the clause vertex $c$ does not belong to $D$, there exists at least one true vertex adjacent to $c$ which belongs to $D$. Suppose to the contrary that $c$ has strictly more than one neighbor in $D$, say $T_{x_1}^i$ and $T_{x_2}^j$ without loss of generality. Then, $\{T_{x_k}^1,T_{x_k}^2,T_{x_k}^3\} \subset D$ for $k=1,2$ which implies that $D \cap \{F_{x_1}^1,F_{x_1}^2,F_{x_1}^3,F_{x_2}^1,F_{x_2}^2,F_{x_2}^3\} = \emptyset$ as $\vert D \cap V(G_{x_k}) \vert = 3$ for $k=1,2$. It follows that the variable vertices $x_1$ and $x_2$ must be dominated by some vertices in $G_c$; but $\vert D \cap V(G_c) \vert = 1$ and $N[x_1] \cap N[x_2] = \emptyset$ and so, either $x_1$ or $x_2$ is not dominated. Thus, $c$ has exactly one neighbor in $D$, say $T_{x_1}^i$ without loss of generality. Then, necessarily $D \cap V(G_c) = \{l_{\{x_1\}}\}$  for otherwise either $x_1$ or some vertex in $K_c$ would not be dominated. But then, it is clear that every vertex in $G_c$ is dominated by exactly one vertex; thus, $D$ is efficient.

Conversely, assume that every minimum dominating set of $G_{\Phi}$ is efficient and consider a minimum dominating set $D$ of $G_{\Phi}$. If for some $x \in X$, $\vert D \cap V(G_x) \vert \geq 4$, then clearly at least one vertex in $G_x$ is dominated by two vertices in $D \cap V(G_x)$. Thus, $\vert D \cap V(G_x) \vert \leq 3$ for any $x \in X$ and we conclude by Observation \ref{obs:size} that in fact, equality holds. The next observation immediately follows from the fact that $D$ is efficient.

\begin{observation}
\label{obs:effgx}
For any $x \in X$, if $\vert D \cap V(G_x) \vert = 3$ then either $\{u_x^1,u_x^2,u_x^3\} \subset D$, $\{T_x^1,T_x^2,T_x^3\} \subset D$ or $\{F_x^1,F_x^2,F_x^3\} \subset D$.
\end{observation}

Now, consider a clause $c \in C$ containing variables $x_1$, $x_2$ and $x_3$ and suppose without loss of generality that $T_{x_1}^1$ is adjacent to $c$ (note that then the variable vertex $x_1$ is adjacent to $F_{x_1}^1$). If the clause vertex $c$ belongs to $D$ then, since $D$ is efficient, $T_{x_1}^1 \notin D$ and $u_{x_1}^1,F_{x_1}^1 \notin D$ ($T_{x_1}^1$ would otherwise be dominated by at least two vertices) which contradicts Observation \ref{obs:effgx}. Thus, no clause vertex belongs to $D$. Similarly, suppose that there exists $i \in \{1,2,3\}$ such that $x_i \in D$, say $x_1 \in D$ without loss of generality. Then, since $D$ is efficient, $F_{x_1}^1 \notin D$ and $T_{x_1}^1,u_{x_1}^2 \notin D$ ($F_{x_1}^1$ would otherwise be dominated by at least two vertices) which again contradicts Observation \ref{obs:effgx}. Thus, no variable vertex belongs to $D$. Finally, since $D$ is efficient, $\vert D \cap V(K_c) \vert \leq 1$ and so, $\vert D \cap V(G_c) \vert = 1$ by Observation \ref{obs:size}.
\end{claimproof}

Now by combining Claims \ref{clm:phisat} and \ref{clm:eff}, we obtain that $\Phi$ is satisfiable if and only if every minimum dominating set of $G_{\Phi}$ is efficient, that is, $G_{\Phi}$ is a \yes-instance for {\sc All Efficient MD}.
\end{proof}

\begin{theorem}
\label{thm:indepmd2}
{\sc All Independent MD} is $\mathsf{NP}$-hard when restricted to subcubic claw-free graphs.
\end{theorem}

\begin{proof}
We use a reduction from {\sc Positive Exactly 3-Bounded 1-In-3 3-Sat}, where each variable appears in exactly three clauses and only positively, each clause contains three positive literals, and we want a truth assignment such that each clause contains exactly one true literal. This problem is shown to be $\mathsf{NP}$-complete in \cite{moore}. Given an instance $\Phi$ of this problem, with variable set $X$ and clause set $C$, we construct an equivalent instance of {\sc All Independent MD} as follows. Consider the graph $G_{\Phi} = (V,E)$ constructed in the proof of Lemma \ref{lemma:efficient} and let $V_i = \{v \in V: d_{G_{\Phi}}(v) = i\}$ for $i =2,3$ (note that no vertex in $G_{\Phi}$ has degree one). Then, for any $v\ \in V_3$, we replace the vertex $v$ by the gadget $G_v$ depicted in Fig. \ref{fig:dv3}; and for any $v \in V_2$, we replace the vertex $v$ by the gadget $G_v$ depicted in Fig. \ref{fig:dv2}. We denote by $G'_{\Phi}$ the resulting graph. Note that $G'_{\Phi}$ is claw-free and $\Delta (G'_{\Phi}) = 3$ (also note that no vertex in $G'_{\Phi}$ has degree one). It is shown in the proof of Lemma \ref{lemma:efficient} that $\Phi$ is satisfiable if and only if $G_{\Phi}$ is a \yes-instance for {\sc All Efficient MD}; we here show that $G_{\Phi}$ is a \yes-instance for {\sc All Efficient MD} if and only if $G'_{\Phi}$ is a \yes-instance for {\sc All Independent MD}. To this end, we first prove the following.

\begin{figure}[htb]
\centering
\begin{subfigure}[b]{.45\textwidth}
\centering
\begin{tikzpicture}[scale=.5]
\node[circ,label=below:{\tiny $v_2$}] (v2) at (0,0) {};
\node[circ,label=below:{\tiny $u_2$}] (u2) at (1,0) {};
\node[circ,label=left:{\tiny $w_2$}] (w2) at (.5,1) {};

\node[circ,label=below:{\tiny $v_3$}] (v3) at (5,0) {};
\node[circ,label=below:{\tiny $w_3$}] (w3) at (4,0) {};
\node[circ,label=right:{\tiny $u_3$}] (u3) at (4.5,1) {};

\node[circ,label=right:{\tiny $v_1$}] (v1) at (2.5,5) {};
\node[circ,label=left:{\tiny $u_1$}] (u1) at (2,4) {};
\node[circ,label=right:{\tiny $w_1$}] (w1) at (3,4) {};

\draw[-] (v2) -- (v1) node[circ,pos=.5,label=left:{\tiny $b_1$}] {} node[circ,pos=.35,label=left:{\tiny $c_1$}] {} node[circ,pos=.65,label=left:{\tiny $a_1$}] {};
\draw[-] (v3) -- (v1) node[circ,pos=.5,label=right:{\tiny $b_3$}] {} node[circ,pos=.35,label=right:{\tiny $a_3$}] {} node[circ,pos=.65,label=right:{\tiny $c_3$}] {};
\draw[-] (v2) -- (v3) node[circ,pos=.5,label=below:{\tiny $b_2$}] {} node[circ,pos=.35,label=below:{\tiny $a_2$}] {} node[circ,pos=.65,label=below:{\tiny $c_2$}] {};
\draw[-] (u1) -- (w1) 
(u2) -- (w2)
(u3) -- (w3);

\draw[-] (-.5,-.5) -- (v2) node[near start,left] {\tiny $e_2$};
\draw[-] (5.5,-.5) -- (v3) node[near start,right] {\tiny $e_3$};
\draw[-] (2.5, 5.5) -- (v1) node[near start,left] {\tiny $e_1$};

\draw[-Implies,line width=.6pt,double distance=2pt] (-1,2.5) -- (0,2.5);

\node[circ,label=below:{\tiny $v$}] (v) at (-3,2.5) {};
\node[circ] (e2) at (-4,1.5) {};
\node[circ] (e3) at (-2,1.5) {};
\node[circ] (e1) at (-3,3.5) {};

\draw[-] (v) -- (e1) node[midway,right] {\tiny $e_1$};
\draw[-] (v) -- (e2) node[midway,left] {\tiny $e_2$};
\draw[-] (v) -- (e3) node[midway,right] {\tiny $e_3$};
\end{tikzpicture}
\caption{$d_{G_{\Phi}}(v) = 3$.}
\label{fig:dv3}
\end{subfigure}
\hspace*{.5cm}
\begin{subfigure}[b]{.45\textwidth}
\centering
\begin{tikzpicture}[scale=.6]
\node[circ,label=above:{\tiny $v$}] (v) at (1,0) {};
\node[circ] (1) at (0,0) {};
\node[circ] (2) at (2,0) {};

\draw[-] (1) -- (v) node[midway,above] {\tiny $e_1$};
\draw[-] (v) -- (2) node[midway,above] {\tiny $e_2$};

\draw[-Implies,line width=.6pt,double distance=2pt] (2.5,0) -- (3.5,0);

\node[circ,label=above:{\tiny $v_1$}] (v1) at (4.5,0) {};
\node[circ,label=above:{\tiny $u_1$}] at (5,0) {};
\node[circ,label=above:{\tiny $a_1$}] at (5.5,0) {};
\node[circ,label=above:{\tiny $b_1$}] at (6,0) {};
\node[circ,label=above:{\tiny $c_1$}] at (6.5,0) {};
\node[circ,label=above:{\tiny $u_2$}] at (7,0) {};
\node[circ,label=above:{\tiny $v_2$}] (v2) at (7.5,0) {};

\draw[-] (4,0) -- (8,0) node[pos=0,above] {\tiny $e_1$} node[pos=1,above] {\tiny $e_2$};

\node[invisible] at (3,-3) {};
\end{tikzpicture}
\caption{$d_{G_{\Phi}}(v) = 2$.}
\label{fig:dv2}
\end{subfigure}
\caption{The gadget $G_v$.}
\end{figure}
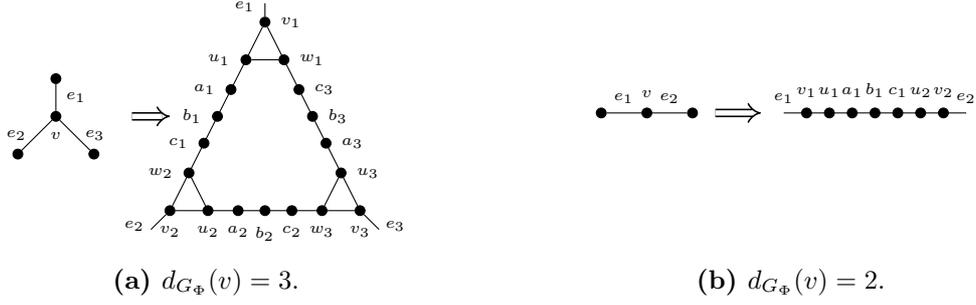

\begin{Claim}
\label{clm:gammas}
$\gamma (G'_{\Phi}) = \gamma (G_{\Phi}) + 5 \vert V_3 \vert + 2 \vert V_2 \vert$.
\end{Claim}

\begin{claimproof}
Let $D$ be a minimum dominating set of $G_{\Phi}$. We construct a dominating set $D'$ of $G'_{\Phi}$ as follows. For any $v \in D$, if $v \in V_3$, add $v_1$, $v_2$, $v_3$, $b_1$, $b_2$, and $b_3$ to $D'$; otherwise, add $v_1$, $v_2$ and $b_1$ to $D'$. For any $v \in V \setminus D$, let $u \in D$ be a neighbor of $v$, say $e_1 =uv$ without loss of generality. Then, if $v \in V_3$, add $a_1$, $c_3$, $w_2$, $u_3$ and $b_2$ to $D'$; otherwise, add $a_1$ and $u_2$ to $D'$. Clearly, $D'$ is dominating and $\vert D' \vert = \gamma (G_{\Phi}) + 5 \vert V_3 \vert + 2 \vert V_2 \vert \geq \gamma (G'_{\Phi})$.

\begin{observation}
\label{obs:size2}
For any dominating set $D'$ of $G'_{\Phi}$, the following holds.
\begin{itemize}
\item[(i)] For any $v \in V_2$, $\vert D' \cap V(G_v) \vert \geq 2$. Moreover, if equality holds then $D' \cap \{v_1,v_2\} = \emptyset$ and there exists $j \in \{1,2\}$ such that $u_j \notin D'$.
\item[(ii)] For any $v \in V_3$, $\vert D' \cap V(G_v) \vert \geq 5$. Moreover, if equality holds then $D' \cap \{v_1,v_2,v_3\} = \emptyset$ and there exists $j \in \{1,2,3\}$ such that $D' \cap \{u_j,v_j,w_j\} = \emptyset$. 
\end{itemize} 
\end{observation}

(i) Clearly, $D' \cap \{v_1,u_1,a_1\} \neq \emptyset$ and $D' \cap \{c_1,u_2,v_2\} \neq \emptyset$ as $u_1$ and $u_2$ must be dominated. Thus, $\vert D' \cap V(G_v) \vert \geq 2$. Now, suppose that $D' \cap \{v_1,v_2\} \neq \emptyset$ say $v_1 \in D'$ without loss of generality. Then $D' \cap \{u_1,a_1,b_1\} \neq \emptyset$ as $a_1$ must be dominated which implies that $\vert D' \cap V(G_v) \vert \geq 3$ (recall that $D' \cap \{c_1,u_2,v_2\} \neq \emptyset$). Similarly, if both $u_1$ and $u_2$ belong to $D'$, then $\vert D' \cap V(G_v) \vert \geq 3$ as $D' \cap \{a_1,b_1,c_1\} \neq \emptyset$ ($b_1$ would otherwise not be dominated).\\

(ii) Clearly, for any $i \in \{1,2,3\}$, $D' \cap \{a_i,b_i,c_i\} \neq \emptyset$ as $b_i$ must be dominated. Now, if there exists $j \in \{1,2,3\}$ such that $D' \cap \{u_j,v_j,w_j\} = \emptyset$, say $j= 1$ without loss of generality, then $a_1, c_3 \in D'$ (one of $u_1$ and $w_1$ would otherwise not be dominated). But then, $D' \cap \{b_1,c_1,w_2\} \neq \emptyset$ as $c_1$ must be dominated, and $D' \cap \{a_3,b_3,u_3\} \neq \emptyset$ as $a_3$ must be dominated; and so, $\vert D' \cap V(G_v) \vert \geq 5$ (recall that $D' \cap \{a_2,b_2,c_2\} \neq  \emptyset$). Otherwise, for any $j \in \{1,2,3\}$, $D' \cap \{u_j,v_j,w_j\} \neq \emptyset$ which implies that $\vert D' \cap V(G_v) \vert \geq 6$.

Now suppose that $D' \cap \{v_1,v_2,v_3\} \neq \emptyset$, say $v_1 \in D'$ without loss of generality. If there exists $j \neq 1$ such that $D' \cap \{u_j,v_j,w_j\} = \emptyset$, say $j = 2$ without loss of generality, then $c_1,a_2 \in D'$ (one of $u_2$ and $w_2$ would otherwise not be dominated). But then, $D' \cap \{a_1,b_1,u_1\} \neq \emptyset$ as $a_1$ should be dominated, and $D' \cap \{b_2,c_2,w_3\} \neq \emptyset$ as $c_2$ must be dominated. Since $D' \cap \{a_3,b_3,c_3\} \neq \emptyset$, it then follows that $\vert D' \cap V(G_v) \vert \geq 6$. Otherwise, $D' \cap \{u_j,v_j,w_j\} \neq \emptyset$ for any $j \in \{1,2,3\}$ and so, $\vert D' \cap V(G_v) \vert \geq 6$ (recall that $D' \cap \{a_i,b_i,c_i\} \neq \emptyset$ for any $i \in \{1,2,3\}$). $\diamond$ 

\begin{observation}
\label{obs:min}
If $D'$ is a minimum dominating set of $G'_{\Phi}$, then $\vert D' \cap V(G_v) \vert \leq 3$ for any $v \in V_2$ and $\vert D' \cap V(G_v) \vert \leq 6$ for any $v \in V_3$.
\end{observation}   

Indeed, if $v \in V_2$ then $\{v_1,b_1,v_2\}$ is a dominating set of $V(G_v)$; and if $v \in V_3$, then $\{v_1,v_2,v_3,b_1,b_2,b_3\}$ is a dominating set of $V(G_v)$. $\diamond$ \\

Now, consider a minimum dominating set $D'$ of $G'_{\Phi}$ and let $D_3 = \{v \in V_3: \vert D' \cap V(G_v) \vert = 6\}$ and $D_2 = \{v \in V_2 : \vert D' \cap V(G_v) \vert = 3\}$. We claim that $D = D_3 \cup D_2$ is a dominating set of $G_{\Phi}$. Indeed, consider a vertex $v \in V \setminus D$. We distinguish two cases depending on whether $v \in V_2$ of $v \in V_3$.\\

\noindent
\textbf{Case 1.} $v \in V_2$. Then $\vert D' \cap V(G_v) \vert = 2$ by construction, which by Observation \ref{obs:size2}(i) implies that there exists $j \in \{1,2\}$ such that $D' \cap \{v_j,u_j\} = \emptyset$ , say $j = 1$ without loss of generality. Since $v_1$ must be dominated, $v_1$ must then have a neighbor $x_i$ belonging to $D'$, for some vertex $x$ adjacent to $v$ in $G_{\Phi}$. But then, it follows from Observation \ref{obs:size2} that $\vert D' \cap V(G_x) \vert > 2$ if $x \in V_2$, and $\vert D' \cap V(G_x) \vert > 5$ if $x \in V_3$ (indeed, $x_i \in D'$); thus, $x \in D$.\\

\noindent
\textbf{Case 2.} $v \in V_3$. Then $\vert D' \cap V(G_v) \vert = 5$ by construction, which by Observation \ref{obs:size2}(ii) implies that there exists $j \in \{1,2,3\}$ such that $D' \cap \{u_j,v_j,w_j\} = \emptyset$, say $j = 1$ without loss of generality. Since $v_1$ must be dominated, $v_1$ must then have a neighbor $x_i$ belonging to $D'$, for some vertex $x$ adjacent to $v$ in $G_{\Phi}$. But then, it follows from Observation \ref{obs:size2} that $\vert D' \cap V(G_x) \vert > 2$ if $x \in V_2$, and $\vert D' \cap V(G_x) \vert > 5$ if $x \in V_3$ (indeed, $x_i \in D'$); thus, $x \in D$.\\

Hence, $D$ is a dominating set of $G_{\Phi}$. Moreover, it follows from Observations \ref{obs:size2} and \ref{obs:min} that $\vert D' \vert = 6 \vert D_3 \vert + 5 \vert V_3 \setminus D_3 \vert + 3 \vert D_2 \vert + 2 \vert V_2 \setminus D_2 \vert = \vert D \vert+ 5 \vert V_3 \vert + 2 \vert V_2 \vert$. Thus, $\gamma (G'_{\Phi}) = \vert D' \vert \geq \gamma (G_{\Phi}) + 5 \vert V_3 \vert + 2 \vert V_2 \vert$ and so, $\gamma (G'_{\Phi}) = \gamma (G_{\Phi}) + 5 \vert V_3 \vert + 2 \vert V_2 \vert$. Finally note that this implies that the constructed dominated set $D$ is in fact minimum.
\end{claimproof}

We next show that $G_{\Phi}$ is a \yes-instance for {\sc All Efficient MD} if and only if $G'_{\Phi}$ is a \yes-instance for {\sc All Independent MD}. Since $\Phi$ is satisfiable if and only if $G_{\Phi}$ is a \yes-instance for {\sc All Efficient MD}, as shown in the proof of Lemma \ref{lemma:efficient}, this would conclude the proof.\\

Assume first that $G_{\Phi}$ is a \yes-instance for {\sc All Efficient MD} and suppose to the contrary that $G'_{\Phi}$ is a \no-instance for {\sc All Independent MD} that is, $G'_{\Phi}$ has a minimum dominating set $D'$ which is not independent. Denote by $D$ the minimum dominating set of $G_{\Phi}$ constructed from $D'$ according to the proof of Claim \ref{clm:gammas}. Let us show that $D$ is not efficient. Consider two adjacent vertices $a,b \in D'$. If $a$ and $b$ belong to gadgets $G_x$ and $G_v$ respectively, for two adjacent vertices $x$ and $v$ in $G_{\Phi}$, that is, $a$ is of the form $x_i$ and $b$ is of the form $v_j$, then by Observation \ref{obs:size2} $x,v \in D$ and so, $D$ is not efficient. Thus, it must be that $a$ and $b$ both belong the same gadget $G_v$, for some $v \in V_2 \cup V_3$. We distinguish cases depending on whether $v \in V_2$ or $v \in V_3$.\\

\noindent
\textbf{Case 1.} $v \in V_2$. Suppose that $\vert D' \cap V(G_v) \vert = 2$. Then by Observation \ref{obs:size2}(i), $D' \cap \{v_1,v_2\} = \emptyset$ and there exists $j \in \{1,2\}$ such that $u_j \notin D'$, say $u_1 \notin D'$ without loss of generality. Then, necessarily $a_1 \in D'$ ($u_1$ would otherwise not be dominated) and so, $b_1 \in D'$ as $D' \cap V(G_v)$ contains an edge and $\vert D' \cap V(G_v) \vert = 2$ by assumption; but then, $u_2$ is not dominated. Thus, $\vert D' \cap V(G_v) \vert \geq 3$ and we conclude by Observation \ref{obs:min} that in fact, equality holds. Note that consequently, $v \in D$. We claim that then, $\vert D' \cap \{v_1,v_2\} \vert \leq 1$. Indeed, if both $v_1$ and $v_2$ belong to $D'$, then $b_1 \in D'$ (since $\vert D' \cap V(G_v) \vert = 3$, $D'$ would otherwise not be dominating) which contradicts that fact that $D' \cap V(G_v)$ contains an edge. Thus, $\vert D' \cap \{v_1,v_2\} \vert \leq 1$ and we may assume without loss of generality that $v_2 \notin D'$. Let $x_i \neq u_2$ be the other neighbor of $v_2$ in $G'_{\Phi}$, where $x$ is a neigbhor of $v$ in $G_{\Phi}$. 

Suppose first that $x \in V_2$. Then, $\vert D' \cap V(G_x) \vert = 2$ for otherwise $x$ would belong to $D$ and so, $D$ would contain the edge $vx$. It then follows from Observation \ref{obs:size2}(i) that there exists $j \in \{1,2\}$ such that $D' \cap \{x_j,y_j\} = \emptyset$, where $y_j$ is the neighbor of $x_j$ in $V(G_x)$. We claim that $j \neq i$; indeed, if $j = i$, since $v_2,x_i,y_i \notin D'$, $x_i$ would not be dominated. But then, $x_j$ must have a neighbor $t_k  \neq y_j$, for some vertex $t$ adjacent to $x$ in $G_{\Phi}$, which belongs to $D'$; it then follows from Observation \ref{obs:size2} and the construction of $D$ that $t \in D$ and so, $x$ has two neighbors in $D$, namely $v$ and $t$, a contradiction.

Second, suppose that $x \in V_3$. Then, $\vert D' \cap V(G_x) \vert = 5$ for otherwise $x$ would belong to $D$ and so, $D$ would contain the edge $vx$. It then follows from Observation \ref{obs:size2}(ii) that there exists $j \in \{1,2,3\}$ such that $D' \cap \{x_j,y_j,z_j\} = \emptyset$, where $y_j$ and $z_j$ are the two neighbors of $x_j$ in $V(G_x)$. We claim that $j \neq i$; indeed, if $j = i$, since $v_2,x_i,y_i,z_i \notin D'$, $x_i$ would not be dominated. But then, $x_j$ must have a neighbor $t_k \neq y_j,z_j$, for some vertex $t$ adjacent to $x$ in $G_{\Phi}$, which belongs to $D'$; it then follows from Observation \ref{obs:size2} and the construction of $D$ that $t \in D$ and so, $x$ has two neighbors in $D$, namely $v$ and $t$, a contradiction.\\

\noindent
\textbf{Case 2.} $v \in V_3$. Suppose that $\vert D' \cap V(G_v) \vert = 5$. Then, by Observation \ref{obs:size2}(ii), $D' \cap \{v_1,v_2,v_3\} = \emptyset$ and there exists $j \in \{1,2,3\}$ such that $D' \cap \{u_j,v_j,w_j\} = \emptyset$, say $j= 1$ without loss of generality. Then, $a_1,c_3 \in D'$ (one of $u_1$ and $w_1$ would otherwise not be dominated), $D' \cap \{c_1,w_2,u_2\} \neq \emptyset$ ($w_2$ would otherwise not be dominated), $D' \cap \{a_3,u_3,w_3\} \neq \emptyset$ ($u_3$ would otherwise not be dominated) and $D' \cap \{a_2,b_2,c_2\} \neq \emptyset$ ($b_2$ would otherwise not be dominated); in particular, $b_1,b_3 \notin D'$ as $\vert D' \cap V(G_v) \vert = 5$ by assumption. Since $D' \cap V(G_v)$ contains an edge, it follows that either $u_2,a_2 \in D'$ or $c_2,w_3 \in D'$; but then, either $c_1$ or $a_3$ is not dominated, a contradiction. Thus, $\vert D' \cap V(G_v) \vert \geq 6$ and we conclude by Observation \ref{obs:min} that in fact, equality holds. Note that consequently, $v \in D$. It follows that $\{v_1,v_2,v_3\} \not\subset D'$ for otherwise $D' \cap V(G_v) = \{v_1,v_2,v_3,b_1,b_2,b_3\}$ and so, $D' \cap V(G_v)$ contains no edge. Thus, we may asssume without loss of generality that $v_1 \notin D'$. Denoting by $x_i \neq u_1,w_1$ the third neighbor of $v_1$, where $x$ is a neighbor of $v$ in $G_{\Phi}$, we then proceed as in the previous case to conclude that $x$ has two neighbors in $D$.\\

Thus, $D$ is not efficient, which contradicts the fact that $G_{\Phi}$ is a \yes-instance for {\sc All Efficient MD}. Hence, every minimum dominating set of $G'_{\Phi}$ is independent i.e., $G'_{\Phi}$ is a \yes-instance for {\sc All Independent MD}.\\

Conversely, assume that $G'_{\Phi}$ is a \yes-instance for {\sc All Independent MD} and suppose to the contrary that $G_{\Phi}$ is a \no-instance for {\sc All Efficient MD} that is, $G_{\Phi}$ has a minimum dominating set $D$ which is not efficient. Let us show that $D$ either contains an edge or can be transformed into a minimum dominating set of $G_{\Phi}$ containing an edge. Since any minimum dominating of $G'_{\Phi}$ constructed according to the proof of Claim~\ref{clm:gammas} from a minimum dominating set of $G_{\Phi}$ containing an edge, also contains an edge, this would lead to a contradiction and thus conclude the proof.

Suppose that $D$ contains no edge. Since $D$ is not efficient, there must then exist a vertex $v \in V \setminus D$ such that $v$ has two neighbors in $D$. We distinguish cases depending on which type of vertex $v$ is.\\

\noindent
\textbf{Case 1.} \textit{$v$ is a variable vertex.} Suppose that $v = x_1$ in some clause gadget $G_c$, where $c \in C$ contains variables $x_1$, $x_2$ and $x_3$, and assume without loss of generality that $x_1$ is adjacent to $F_{x_1}^1$. By assumption, $F_{x_1}^1,l_{\{x_1\}} \in D$ which implies that $D \cap \{l_{\{x_2\}},l_{\{x_3\}},T_{x_1}^1,u_{x_1}^2\} = \emptyset$ ($D$ would otherwise contain an edge). We may then assume that $F_{x_2}^i$ and $F_{x_3}^j$, where $F_{x_2}^ix_2, F_{x_3}^jx_3 \in E(G_{\Phi})$, belong to $D$; indeed, since $x_2$ (resp. $x_3$) must be dominated, $D \cap \{F_{x_2}^i,x_2\} \neq \emptyset$ (resp. $D \cap \{F_{x_3}^j,x_3\} \neq \emptyset$) and since $l_{\{x_1\}} \in D$, $(D \setminus \{x_2\}) \cup \{F_{x_2}^i\}$ (resp. $(D \setminus \{x_3\}) \cup \{F_{x_3}^j\}$) remains dominating. We may then assume that $T_{x_2}^i,T_{x_3}^j \notin D$ for otherwise $D$ would contain an edge. It follows that $c \in D$ ($c$ would otherwise not be dominated); but then, it suffices to consider $(D \setminus \{c\}) \cup \{T_{x_1}^1\}$ to obtain a minimum dominating set of $G_{\Phi}$ containing an edge.\\

\noindent
\textbf{Case 2.} \textit{$v = u_x^i$ for some variable $x \in X$ and $i \in \{1,2,3\}$.} Assume without loss of generality that $i = 1$. Then $T_x^1,F_x^3 \in D$ by assumption, which implies that $F_x^1,T_x^3 \notin D$ ($D$ would otherwise contain an edge). But then, $\vert D \cap \{u_x^2,F_x^2,T_x^2,u_x^3\} \vert \geq 2$ as $u_x^2$ and $u_x^3$ must be dominated; and so, $(D \setminus \{u_x^3,F_x^2,T_x^2,u_x^2\}) \cup \{F_x^2,T_x^2\}$ is a dominating set of $G_{\Phi}$ of size at most that of $D$ which contains an edge.\\

\noindent
\textbf{Case 3.} \textit{$v$ is a clause vertex.} Suppose that $v = c$ for some clause $c \in C$ containing variables $x_1$, $x_2$ and $x_3$, and assume without loss of generality that $c$ is adjacent to $T_{x_i}^1$ for any $i \in \{1,2,3\}$. By assumption $c$ has two neighbors in $D$, say $T_{x_1}^1$ and $T_{x_2}^1$ without loss of generality. Since $D$ contains no edge, it follows that $F_{x_1}^1,F_{x_2}^1 \notin D$; but then, $\vert D \cap \{x_1,x_2,l_{\{x_1\}},l_{\{x_2\}}\} \vert \geq 2$ (one of $x_1$ and $x_2$ would otherwise not be dominated) and so, $(D \setminus \{x_1,x_2,l_{\{x_1\}},l_{\{x_2\}}\}) \cup \{l_{\{x_1\}},l_{\{x_2\}}\}$ is a dominating set of $G_{\Phi}$ of size at most that of $D$ which contains an edge.\\

\noindent
\textbf{Case 4.} \textit{$v \in V(K_c)$ for some clause $c \in C$.} Denote by $x_1$, $x_2$ and $x_3$ the variables contained in $c$ and assume without loss of generality that $v = l_{\{x_1\}}$. Since $l_{\{x_1\}}$ has two neighbors in $D$ and $D$ contains no edge, necessarily $x_1 \in D$. Now assume without loss of generality that $x_1$ is adjacent to $F_{x_1}^1$ (note that by construction, $c$ is then adjacent to $T_{x_1}^1$). Then, $F_{x_1}^1 \notin D$ ($D$ would otherwise contain an edge) and $T_{x_1}^1,u_{x_1}^2 \notin D$  for otherwise $(D \setminus \{x_1\}) \cup \{F_{x_1}^1\}$ would be a minimum dominating set of $G_{\Phi}$ containing an edge (recall that by assumption, $D \cap V(K_c) \neq \emptyset$). It follows that $T_{x_1}^2 \in D$ ($u_{x_1}^2$ would otherwise not be dominated) and so, $F_{x_1}^2 \notin D$ as $D$ contains no edge. It follows that $\vert D \cap \{u_{x_1}^1,F_{x_1}^3,T_{x_1}^3,u_{x_1}^3\} \vert \geq 2$ as $u_{x_1}^1$ and $u_{x_1}^3$ must be dominated. Now if $c$ belongs to $D$, then $(D \setminus \{u_{x_1}^1,F_{x_1}^3,T_{x_1}^3,u_{x_1}^3\}) \cup \{F_{x_1}^3,T_{x_1}^3\}$ is a dominating set of $G_{\Phi}$ of size at most that of $D$ which contains an edge. Thus, we may assume that $c \notin D$ which implies that $u_{x_1}^1 \in D$ ($T_{x_1}^1$ would otherwise not be dominated) and that there exists $j \in \{2,3\}$ such that $T_{x_j}^i \in D$ with $cT_{x_j}^i \in E(G_{\Phi})$ ($c$ would otherwise not be dominated). Now, since $u_{x_1}^3$ must be dominated and $F_{x_1}^2 \notin D$, it follows that $D \cap \{u_{x_1}^3,T_{x_1}^3\} \neq \emptyset$ and we may assume that in fact $T_{x_1}^3 \in D$ (recall that $T_{x_1}^2 \in D$ and so, $F_{x_1}^2$ is dominated). But then, by considering the minimum dominating set $(D \setminus \{u_{x_1}^1\}) \cup \{T_{x_1}^1\}$, we fall back into Case 3 as $c$ is then dominated by both $T_{x_1}^1$ and $T_{x_j}^i$.\\

\noindent
\textbf{Case 5.} \textit{$v$ is a true vertex.} Assume without loss of generality that $v = T_x^1$ for some variable $x \in X$. Suppose first that $u_x^1 \in D$. Then since $D$ contains no edge, $F_x^3 \notin D$; furthermore, denoting by $t \neq u_x^1,T_x^3$ the variable vertex adjacent to $F_x^3$, we also have $t \notin D$ for otherwise $(D \setminus \{u_x^1\}) \cup \{F_x^3\}$ would be a minimum dominating set containing an edge (recall that $T_x^1$ has two neighbors in $D$ by assumption). But then, since $t$ must be dominated, it follows that the second neighbor of $t$ must belong to $D$; and so, by considering the minimum dominating set  $(D \setminus \{u_x^1\}) \cup \{F_x^3\}$, we fall back into Case 1 as the variable vertex $t$ is then dominated by two vertices. Thus, we may assume that $u_x^1 \notin D$ which implies that $F_x^1,c \in D$, where $c$ is the clause vertex adjacent to $T_x^1$. Now, denote by $x_1 = x$, $x_2$ and $x_3$ the variables contained in $c$ (note that by construction, $x_1$ is then adjacent to $F_{x_1}^1$). Then, $x_1 \notin D$ ($D$ would otherwise contain the edge $F_{x_1}^1x_1$) and we may assume that $l_{\{x_1\}} \notin D$ (we otherwise fall back into Case 1 as $x_1$ would then have two neighbors in $D$). It follows that $D \cap V(K_c) \neq \emptyset$ ($l_{\{x_1\}}$ would otherwise not be dominated) and since $D$ contains no edge, in fact $\vert D \cap V(K_c) \vert = 1$, say $l_{\{x_2\}} \in D$ without loss of generality. Then, $x_2 \notin D$ as $D$ contains no edge and we may assume that $F_{x_2}^j \notin D$, where $F_{x_2}^j$ is the false vertex adjacent to $x_2$, for otherwise we fall back into Case 1. In the following, we assume without loss of generality that $j = 1$, that is, $x_2$ is adjacent to $F_{x_2}^1$ (note that by construction, $c$ is then adjacent to $T_{x_2}^1$). Now, since the clause vertex $c$ belongs to $D$ by assumption, it follows that $T_{x_2}^1 \notin D$ ($D$ would otherwise contain the edge $cT_{x_2}^1$); and as shown previously, we may assume that $u_{x_2}^1 \notin D$ (indeed, $T_{x_2}^1$ would otherwise have two neighbors in $D$, namely $c$ and $u_{x_2}^1$, but this case has already been dealt with). Then, since $u_{x_2}^1$ and $F_{x_2}^1$ must be dominated, necessarily $F_{x_2}^3$ and $u_{x_2}^2$ belong to $D$ (recall that $D \cap \{x_2,F_{x_2}^1,T_{x_2}^1,u_{x_2}^1 \} = \emptyset$) which implies that $T_{x_2}^3,T_{x_2}^2 \notin D$ ($D$ would otherwise contain an edge). Now since $u_{x_2}^3$ must be dominated, $D \cap \{u_{x_2}^3,F_{x_2}^2\} \neq \emptyset$ and we may assume without loss of generality that in fact, $F_{x_2}^2 \in D$. But then, by considering the minimum dominating set $(D \setminus \{u_{x_2}^2\}) \cup \{F_{x_2}^1\}$, we fall back into Case~1 as $x_2$ is then dominated by two vertices.\\

\noindent
\textbf{Case 6.} \textit{$v$ is a false vertex.} Assume without loss of generality that $v = F_{x_1}^1$ for some variable $x_1 \in X$ and let $c \in C$ be the clause whose corresponding clause vertex is adjacent to $T_{x_1}^1$. Denote by $x_2$ and $x_3$ the two other variables contained in $c$. Suppose first that $x_1 \in D$. Then, we may assume that $D \cap V(K_c) = \emptyset$ for otherwise either $D$ contains an edge (if $l_{\{x_1\}} \in D$) or we fall back into Case 4 ($l_{\{x_1\}}$ would indeed have two neighbors in $D$). Since every vertex of $K_c$ must be dominated, it then follows that $x_2,x_3 \in D$; but then, by considering the minimum dominating set $(D \setminus \{x_1\}) \cup \{l_{\{x_1\}}\}$ (recall that $F_{x_1}^1$ has two neighbors in $D$ by assumption), we fall back into Case 4 as $l_{\{x_2\}}$ is then dominated by two vertices. Thus, we may assume that $x_1 \notin D$ which implies that $T_{x_1}^1,u_{x_1}^2 \in D$ and $T_{x_1}^2,u_{x_1}^1 \notin D$ as $D$ contains no edge. Now, denote by $c'$ the clause vertex adjacent to $T_{x_1}^2$. Then, we may assume that $c' \notin D$ for otherwise we fall back into Case 5 ($T_{x_1}^2$ would indeed have two neighbors in $D$); but then, there must exist a true vertex, different from $T_{x_1}^2$, adjacent to $c'$ and belonging to $D$ ($c'$ would otherwise not be dominated) and by considering the minimum dominating set $(D \setminus \{u_{x_1}^2\}) \cup \{T_{x_1}^2\}$, we then fall back into Case 3 ($c'$ would indeed be dominated by two vertices).\\

Consequently, $G_{\Phi}$ has a minimum dominating set which is not independent which implies that $G'_{\Phi}$ also has a minimum dominating set which is not independent, a contradiction which concludes the proof.  
\end{proof}

Theorem \ref{thm:clawfree} now easily follows from Theorem \ref{thm:indepmd2} and Fact \ref{obs:equi}.


\section{The proof of Theorem \ref{thm:2p3free}}

\setcounter{observation}{0}
\setcounter{Claim}{0}

In this section, we show that \contracd{} is $\mathsf{coNP}$-hard when restricted to $2P_3$-free graphs. To this end, we prove the following.

\begin{theorem}
\label{thm:indp7free}
{\sc All Independent MD} is $\mathsf{NP}$-hard when restricted to $2P_3$-free graphs.
\end{theorem}

\begin{proof}
We reduce from {\sc 3-Sat}: given an instance $\Phi$ of this problem, with variable set $X$ and clause set $C$, we construct an equivalent instance of {\sc All Independent MD} as follows. For any variable $x \in X$, we introduce a copy of $C_3$, which we denote by $G_x$, with one distinguished \textit{positive literal vertex} $x$ and one distinguished \textit{negative literal vertex} $\bar{x}$; in the following, we denote by $u_x$ the third vertex in $G_x$. For any clause $c \in C$, we introduce a \textit{clause vertex} $c$; we then add an edge between $c$ and the (positive or negative) literal vertices whose corresponding literal occurs in $c$. Finally, we add an edge between any two clause vertices so that the set of clause vertices induces a clique denoted by $K$ in the following. We denote by $G_{\Phi}$ the resulting graph.

\begin{observation}
\label{obs:size4}
For any dominating set $D$ of $G_{\Phi}$ and any variable $x \in X$, $\vert D \cap V(G_x) \vert \geq 1$. In particular, $\gamma (G_{\Phi}) \geq \vert X \vert$.
\end{observation}


\begin{Claim}
\label{clm:phisat4}
$\Phi$ is satisfiable if and only if $\gamma (G_{\Phi}) = \vert X \vert$.
\end{Claim}

\begin{claimproof}
Assume that $\Phi$ is satisfiable and consider a truth assignment satisfying $\Phi$. We construct a dominating set $D$ of $G_{\Phi}$ as follows. For any variable $x \in X$, if $x$ is true, add the positive literal vertex $x$ to $D$; otherwise, add the negative variable vertex $\bar{x}$ to $D$. Clearly, $D$ is dominating and we conclude by Observation \ref{obs:size4} that $\gamma (G_{\Phi}) = \vert X \vert$.

Conversely, assume that $\gamma (G_{\Phi}) = \vert X \vert$ and consider a minimum dominating set $D$ of $G_{\Phi}$. Then by Observation \ref{obs:size4}, $\vert D \cap V(G_x) \vert = 1$ for any $x \in X$. It follows that $D \cap K = \emptyset$ and so, every clause vertex must be adjacent to some (positive or negative) literal vertex belonging to $D$. We thus construct a truth assignment satisfying $\Phi$ as follows: for any variable $x \in X$, if the positive literal vertex $x$ belongs to $D$, set $x$ to true; otherwise, set $x$ to false. 
\end{claimproof}

\begin{Claim}
\label{clm:indep2}
$\gamma (G_{\Phi}) = \vert X \vert$ if and only if every minimum dominating set of $G_{\Phi}$ is independent.
\end{Claim}

\begin{claimproof}
Assume that $\gamma (G_{\Phi}) = \vert X \vert$ and consider a minimum dominating set $D$ of $G_{\Phi}$. Then by Observation \ref{obs:size4}, $\vert D \cap V(G_x) \vert = 1$ for any $x \in X$. It follows that $D \cap K = \emptyset$ and since $N[V(G_x)] \cap N[V(G_{x'})] \subset K$ for any two $x,x' \in X$, $D$ is independent.

Conversely, consider a minimum dominating set $D$ of $G_{\Phi}$. Since $D$ is independent, $\vert D \cap V(G_x) \vert \leq 1$ for any $x \in X$ and we conclude by Observation \ref{obs:size4} that in fact, equality holds. Now suppose that there exists $c \in C$, containing variables $x_1$, $x_2$ and $x_3$, such that the corresponding clause vertex $c$ belongs to $D$ (note that since $D$ is independent, $\vert D \cap K \vert \leq 1$). Assume without loss of generality that $x_1$ occurs positively in $c$, that is, $c$ is adjacent to the positive literal vertex $x_1$. Then, $x_1 \notin D$ since $D$ is independent and so, either $u_{x_1} \in D$ or $\bar{x_1} \in D$. In the first case, we immediately obtain that $(D \setminus \{u_{x_1}\}) \cup \{x_1\}$ is a minimum dominating set of $G_{\Phi}$ containing an edge, a contradiction. In the second case, since $c \in D$, any vertex dominated by $\bar{x_1}$ is also dominated by $c$; thus, $(D \setminus \{\bar{x_1}\}) \cup \{x_1\}$ is a minimum dominating set of $G_{\Phi}$ containing an edge, a contradiction. Consequently, $D \cap K = \emptyset$ and so, $\gamma (G_{\Phi}) = \vert D \vert = \vert X \vert$.
\end{claimproof}

Now by combining Claims \ref{clm:phisat4} and \ref{clm:indep2}, we obtain that $\Phi$ is satisfiable if and only if every minimum dominating set of $G_{\Phi}$ is independent, that is, $G_{\Phi}$ is a \yes-instance for {\sc All Independent MD}. There remains to show that $G_{\Phi}$ is $2P_3$-free. To see this, it suffices to observe that any induced $P_3$ of $G_{\Phi}$ contains at least one vertex in the clique $K$. This concludes the proof.
\end{proof}

Theorem \ref{thm:2p3free} now easily follows from Theorem \ref{thm:indp7free} and Fact \ref{obs:equi}.


\section{Conclusion}

In this work, we settle a complexity dichotomy for \contracd{} on $H$-free graphs when $H$ is a connected graph. If we do not require $H$ to be connected, there only remains to settle the complexity status of \contracd{} restricted to $H$-free graphs for $H=P_3+qP_2+pK_1$, when $q\geq 1$ and $p\geq 0$.


 \bibliographystyle{siam}

  \bibliography{bibliography}

\begin{thebibliography}{10}

\bibitem{BTT11}
{\sc C.~Bazgan, S.~Toubaline, and Z.~Tuza}, {\em The most vital nodes with
  respect to independent set and vertex cover}, Discrete Applied Mathematics,
  159 (2011), pp.~1933--1946.

\bibitem{bazgan2013critical}
{\sc C.~Bazgan, S.~Toubaline, and D.~Vanderpooten}, {\em Critical edges for the
  assignment problem: Complexity and exact resolution}, Operations Research
  Letters, 41 (2013), pp.~685--689.

\bibitem{RBPDCZ10}
{\sc C.~Bentz, C.~Marie-Christine, D.~de~Werra, C.~Picouleau, and B.~Ries},
  {\em Blockers and transversals in some subclasses of bipartite graphs: when
  caterpillars are dancing on a grid}, Discrete Mathematics, 310 (2010),
  pp.~132 -- 146.

\bibitem{Bentz}
\leavevmode\vrule height 2pt depth -1.6pt width 23pt, {\em Weighted
  Transversals and Blockers for Some Optimization Problems in Graphs}, Progress
  in Combinatorial Optimization, ISTE-WILEY, 2012, pp.~203--222.

\bibitem{CWP11}
{\sc M.-C. Costa, D.~de~Werra, and C.~Picouleau}, {\em Minimum d-blockers and
  d-transversals in graphs}, Journal of Combinatorial Optimization, 22 (2011),
  pp.~857--872.

\bibitem{Di05}
{\sc R.~Diestel}, {\em Graph Theory}, vol.~173 of Graduate Texts in
  Mathematics, Springer, Heidelberg; New York, fourth~ed., 2010.

\bibitem{DPPR15}
{\sc {\"O}.~Y. Diner, D.~Paulusma, C.~Picouleau, and B.~Ries}, {\em Contraction
  blockers for graphs with forbidden induced paths}, in Algorithms and
  Complexity, Springer International Publishing, 2015, pp.~194--207.

\bibitem{diner2018contraction}
{\sc {\"O}.~Y. Diner, D.~Paulusma, C.~Picouleau, and B.~Ries}, {\em Contraction
  and deletion blockers for perfect graphs and {H}-free graphs}, Theoretical
  Computer Science, 746 (2018), pp.~49 -- 72.

\bibitem{contracdom}
{\sc E.~Galby, P.~T. Lima, and B.~Ries}, {\em Reducing the domination number of
  graphs via edge contractions}, in Mathematical Foundations of Computer
  Science (MFCS) 2019 (to appear), 2019.

\bibitem{HX10}
{\sc J.~Huang and J.-M. Xu}, {\em Domination and total domination contraction
  numbers of graphs}, Ars Combinatoria, 94 (2010).

\bibitem{keller2018blockers}
{\sc C.~Keller and M.~A. Perles}, {\em Blockers for simple hamiltonian paths in
  convex geometric graphs of even order}, Discrete \& Computational Geometry,
  60 (2018), pp.~1--8.

\bibitem{keller2013blockers}
{\sc C.~Keller, M.~A. Perles, E.~Rivera-Campo, and V.~Urrutia-Galicia}, {\em
  Blockers for noncrossing spanning trees in complete geometric graphs}, in
  Thirty Essays on Geometric Graph Theory, Springer, 2013, pp.~383--397.

\bibitem{PBP}
{\sc F.~Mahdavi~Pajouh, V.~Boginski, and E.~Pasiliao}, {\em Minimum vertex
  blocker clique problem}, Networks, 64 (2014), pp.~48--64.

\bibitem{moore}
{\sc C.~Moore and J.~M. Robson}, {\em Hard tiling problems with simple tiles},
  Discrete Computational Geometry, 26 (2001), pp.~573--590.

\bibitem{nasirian2019exact}
{\sc F.~Nasirian, F.~M. Pajouh, and J.~Namayanja}, {\em Exact algorithms for
  the minimum cost vertex blocker clique problem}, Computers \& Operations
  Research, 103 (2019), pp.~296--309.

\bibitem{pajouh2015minimum}
{\sc F.~M. Pajouh, J.~L. Walteros, V.~Boginski, and E.~L. Pasiliao}, {\em
  Minimum edge blocker dominating set problem}, European Journal of Operational
  Research, 247 (2015), pp.~16--26.

\bibitem{PPR16}
{\sc D.~Paulusma, C.~Picouleau, and B.~Ries}, {\em Reducing the clique and
  chromatic number via edge contractions and vertex deletions}, in ISCO 2016,
  vol.~9849 of LNCS, 2016, pp.~38--49.

\bibitem{paulusma2017blocking}
\leavevmode\vrule height 2pt depth -1.6pt width 23pt, {\em Blocking independent
  sets for {H}-free graphs via edge contractions and vertex deletions}, in TAMC
  2017, vol.~10185 of LNCS, 2017, pp.~470--483.

\bibitem{paulusma2018critical}
{\sc D.~Paulusma, C.~Picouleau, and B.~Ries}, {\em Critical vertices and edges
  in {H}-free graphs}, Discrete Applied Mathematics, 257 (2019), pp.~361 --
  367.

\end{thebibliography}

\end{document}